\documentclass[11pt]{article}

\usepackage[letterpaper,margin=1in]{geometry}

\usepackage{amsmath,amssymb,amsthm}
\usepackage{graphicx}

\usepackage{microtype}

\usepackage[hidelinks]{hyperref}
\hypersetup{pdftitle={Matroid Packing Games: The Core and the Nucleolus},
  pdfauthor={Pengfei Liu, Han Xiao, Xin Chen, Qizhi Fang}}

\newtheorem{theorem}{Theorem}
\newtheorem{lemma}{Lemma}
\newtheorem{corollary}{Corollary}

\newtheorem{definition}{Definition}
\newtheorem{remark}{Remark}

\begin{document}

\begin{center}
{\LARGE\bfseries Matroid Packing Games: The Core and the Nucleolus\par}
\bigskip
Pengfei Liu\textsuperscript{1}\quad
Han Xiao\textsuperscript{1}\quad
Xin Chen\textsuperscript{2}\quad
Qizhi Fang\textsuperscript{1,*}\par
\medskip
{\small
\textsuperscript{1}School of Mathematical Sciences, Ocean University of China\\
238 Songling Road, Laoshan District, Qingdao, Shandong 266100, China\\
\textsuperscript{2}Faculty of Information Science and Engineering, Ocean University of China\\
238 Songling Road, Laoshan District, Qingdao, Shandong 266100, China\\
\textsuperscript{*}Corresponding author: \texttt{qfang@ouc.edu.cn}\par}
\end{center}
\vspace{0.5em}

\begin{abstract}
In this paper, we study the matroid base packing problem from the perspective of cooperative game theory and define the associated matroid base packing game. This model extends the network strength game \cite{BB20} from graphic matroids to general matroids. Building on the exchange structure of matroid, we develop a
graph theoretic framework for analyzing the core and the nucleolus in this game. This  framework yields a compact characterization of core nonemptiness and enables polynomial time algorithms for testing core nonemptiness, finding a core allocation, and computing the nucleolus. In particular, the approach provides a broader and more structured framework  for the nucleolus computation problem arising in network strength games \cite{BB20}, covering both the nonempty core and empty core cases.

\end{abstract}
\begin{quote}
\small
\noindent\textbf{Keywords.}
Cooperative game;  matroid packing; core; nucleolus.
\end{quote}
\vspace{0.5em}
\section{Introduction}

The matroid base packing problem, finding the maximum number of disjoint bases in a matroid, is a classical topic in combinatorial optimization with wide applications in network design, resource allocation, and system reliability. In this paper, we examine this problem from a cooperative game theoretic perspective and introduce the matroid packing game. In this game, the players are the elements of the matroid, and the value of a coalition is the maximum number of pairwise disjoint bases contained in the coalition. In particular, the matroid packing game generalizes the network strength game studied in~\cite{BB20}, which corresponds to the special case of graphic matroids.

For example, in a linear matroid the players are vectors and a base is a basis of the span of all available vectors. A coalition can therefore provide several disjoint complete sets of independent resources. In a transversal matroid, assuming that all tasks can be assigned, bases are teams that can be matched to all tasks. Partition matroids similarly model resource bundles meeting prescribed category quotas.

A central question in cooperative game theory is to distribute the total profit to
the players. Many solution concepts have been proposed for profit allocation. One solution concept is the core, which requires that no coalition benefits by breaking
away from the grand coalition. A key algorithmic framework for analyzing the core was established by \cite{DIN99}. Following this work, the core has been studied for flow games~\cite{FZCD02,KP09}, directed acyclic graph games~\cite{BTT17}, network strength games~\cite{BB20}, and arboricity games~\cite{XF23}, among others.
However, the core may be empty in many games. To address this, the least core\cite{MPS79} is introduced as a relaxation. It maximizes the smallest excess over all coalitions, thereby identifying allocations that are "as stable as possible" when exact stability cannot be achieved.

Beyond the least core, the nucleolus is introduced by \cite{Schmeidler}, which is
the unique solution that lexicographically maximizes the vector of nondecreasingly
ordered excess. The standard procedure for computing the nucleolus, called Maschler's scheme and proposed in~\cite{MPS79}, involves solving a sequence of linear programs. However, the size of these linear programs may be exponentially
large due to the number of constraints corresponding to all possible coalitions. Hence
it is in general unclear how to apply this procedure. The first polynomial algorithm
for computing the nucleolus was proposed by \cite{M78} for cooperative cost
games defined on directed trees. Later on, polynomial time algorithms were developed for several classes of combinatorial games, including bankruptcy games~\cite{AM85}, standard tree games~\cite{GMQZ96}, matching games~\cite{KP03,BKP12,KPT20}, airport profit games~\cite{BITZ06}, flow games~\cite{DFS06,PRB06,KP09}, voting games~\cite{EP09}, spanning connectivity games~\cite{ALPS09}, vertex cover games~\cite{CLZ12}, directed acyclic graph games~\cite{BTT17}, shortest path games~\cite{BB19}, network strength games~\cite{BB20}, and arboricity games~\cite{XF23}. On
the negative side, NP hardness results for computing the nucleolus were shown for,
e.g., minimum spanning tree games \cite{UWJ98}, threshold games \cite{EGG07}, $b$ matching games
\cite{KTZ21}, flow games and linear production games \cite{DFS06} and \cite{KP09}.

The main contribution of this paper is twofold. First, we characterize the core and the least core of
the matroid packing game and give polynomial time algorithms for computing the nucleolus in all cases, in the independence oracle model.
Second, we derive from the exchange properties of integer and fractional matroid base packings a
partition and order structure that reduces the exponentially many coalition
constraints in Maschler's scheme to polynomially many local constraints. The closest related works are the studies of network strength games~\cite{BB20}
and arboricity games~\cite{XF23}, both of which also use such a structure:
the former relies on spanning tree packings, network strength computations and graph contractions to construct prime sets and a precedence relation, while the latter uses forest covers and minimal densest minors to obtain a prime partition and an ancestor relation. Arboricity games concern covering, and their prime partition is inspired by the principal partition of matroids. Both works compute the nucleolus when the core is nonempty. Once a partition and an order are available, we likewise reduce coalition constraints to local constraints and use the hierarchy partition. Our construction obtains this structure directly from the exchange graph of a fixed packing. We also establish the local constraint reduction when the core is empty.

For the nonempty core case, we start from any fixed maximum base packing. The exchange graph characterizes the core through zero allocations on uncovered elements and nonnegative differences along its arcs. The relevant strongly connected components (SCCs) are those from which no uncovered element is reachable, and their partial order is given by reachability. This structure gives a compact second linear program whose unique optimal allocation is determined by the hierarchy partition.

When the core is empty, we fix a maximum fractional base packing. The least core is identified with a scaled optimal solution set of the fractional packing dual. A key step is to show that suitable multiples of the packing load vector are convex combinations of incidence vectors of bases and of unions of a maximum number of disjoint bases. These decompositions, together with complementary slackness, connect the fractional packing to tight coalitions corresponding to maximum base packings. We then obtain the least core partition and order from the SCCs of the fractional packing exchange graph, and prove that the second linear program can again be replaced by local constraints. The resulting hierarchy determines the nucleolus. Here fractional packing is an auxiliary optimization problem: the coalition values remain integer packing numbers. This differs from the fractional game considered in~\cite[Section~8]{BB20}, where the characteristic function itself is defined by fractional spanning tree packing.

 This paper is organized as follows. Section 2 introduces relevant concepts and reviews results on matroid exchange properties. Section 3 studies properties
of the core. Section 4 develops an efficient algorithm for computing the
nucleolus when the core is nonempty. Section 5  studies properties
of the  least core and develops an efficient algorithm for computing the
nucleolus when the core is empty.
\section{Preliminaries and Definitions}

\subsection{Cooperative game}

A \emph{cooperative revenue game} $\Gamma=(N,\gamma)$ consists of a \emph{player} set $N$ and a \emph{characteristic function} $\gamma:2^{N}\rightarrow\mathbb{R}$ with $\gamma(\emptyset)=0$. For each subset $S\subseteq N$, called a \emph{coalition}, the value $\gamma(S)$ is the revenue that the players in $S$ can obtain cooperatively. A central question in cooperative game theory is how to distribute the total revenue $\gamma(N)$ among the players in a fair and stable way. An \emph{allocation} of the game $\Gamma=(N,\gamma)$ is a vector $x=(x_{i})_{i\in N}\in\mathbb{R}^{|N|}_{\geq 0}$ with $\sum_{i\in N}x_{i}=\gamma(N)$. Let $X(\Gamma)$ denote the set of allocations of $\Gamma$.
Here and throughout this work, we use the shorthand notation
$
x(S)=\sum_{i\in S}x_i.
$

The \emph{core} of $\Gamma$, denoted by $\mathcal{C}(\Gamma)$, is defined as
\[
\mathcal{C}(\Gamma)=\{x\in X(\Gamma):x(S)\geq\gamma(S),\ \forall S\subseteq N\}.
\]
A core allocation means that no coalition can benefit by breaking away from the grand coalition $N$.  Since the core may be empty or contain many allocations, refined solution concepts, such as the least core and the nucleolus, have been further investigated.

Given an allocation $x\in X(\Gamma)$, the \emph{excess} of a coalition $S\subseteq N$ \emph{w.r.t.} $x$ is defined as
\[
e(S,x)=x(S)-\gamma(S).
\]
It measures the satisfaction of the coalition under $x$. Clearly, $x\in\mathcal{C}(\Gamma)$ if and only if $e(S,x)\geq 0$ for each $S\subseteq N$.

As a relaxation of the core, the \emph{least core} of $\Gamma$ is defined through the following optimization problem.
Let
$\varepsilon_1^*
=
\max\{\varepsilon: x\in X(\Gamma),\
x(S)\ge \gamma(S)+\varepsilon,\ \forall S\subseteq N\}.
$
The least core is
\[
LC(\Gamma)
=
\{x\in X(\Gamma):x(S)\ge \gamma(S)+\varepsilon_1^*,\ \forall S\subseteq N\}.
\]

For the nucleolus, we order all the excesses $e(S,x)$ with $\emptyset\neq S\subsetneq N$ into a nondecreasing sequence and obtain the excess vector
\[
\theta(x)=(e(S_{1},x),e(S_{2},x),\ldots,e(S_{2^{|N|}-2},x)).
\]
The \emph{nucleolus} of $\Gamma$, denoted by $\eta(\Gamma)$, is defined as the set of allocations that lexicographically maximize $\theta(x)$ over $X(\Gamma)$. That is,
\[
\eta(\Gamma)=\{x\in X(\Gamma):\ \theta(x)\geq_{\textrm{lex}}\theta(y)\ \forall y\in X(\Gamma)\}.
\]
It has been shown that $\eta(\Gamma)$ is a singleton, $\eta(\Gamma)\in\mathcal{C}(\Gamma)$ when $\mathcal{C}(\Gamma)\neq\emptyset$, and $\eta(\Gamma)\in\mathcal{L}\mathcal{C}(\Gamma)$ when $\mathcal{C}(\Gamma)=\emptyset$.

Maschler~\cite{MPS79} proposed that $\eta(\Gamma)$ can be obtained by recursively solving a sequence of linear programs, denoted by SLP$(\eta(\Gamma))$. For $k=1,2,\ldots$, the linear program at step $k$ is
\begin{equation}\label{nucleolus LP}
(\textnormal{LP}_k): \ \begin{array}{rl}
\mbox{max}&  \varepsilon\\
\vspace*{1mm} \mbox{s.t.}&\left\{
\begin{array}{ll}
 x(S)\geq \gamma(S)+\varepsilon
 & \quad \forall S\in 2^N\backslash
 \displaystyle\textrm{fix}(X_{k-1}),  \vspace*{1mm}\\
 x\in X_{k-1}. &
\end{array}\right.
\end{array}
\end{equation}
Initially, set $X_{0}=X(\Gamma)$, $\varepsilon_{0}=0$, and
$\textrm{fix}(X_{0})=\{\emptyset,N\}$.
Let $\varepsilon_{k}$ be the optimal value of $(\textnormal{LP}_{k})$, and
$
X_{k}
=
\{x\in X_{k-1}:\,(x,\varepsilon_{k})\ \textrm{is an optimal solution of}\ (\textnormal{LP}_{k})\}.
$
Moreover,
$
\textrm{fix}(X_{k})
=
\{S\subseteq N:\ x(S)=x^{'}(S),\ \forall x,x^{'}\in X_{k}\}.
$
Clearly, the set $X_{1}$ is the least core of $\Gamma$.

\subsection{Matroid packing game}

A pair $\mathcal{M}=(N,\mathcal{I})$ is called a matroid, if $N$ is a finite set and $\mathcal{I}$ is a nonempty collection of subsets of $N$ satisfying:
\begin{enumerate}
\item[(1)] If $I\in\mathcal{I}$ and $J\subseteq I$, then $J\in\mathcal{I}$;
\item[(2)] If $I,J\in\mathcal{I}$ and $|I|<|J|$, then $I+v\in\mathcal{I}$ for some $v\in J\setminus I$.
\end{enumerate}
We use the shorthand notation $I+v=I\cup\{v\}$ and $I-v=I\setminus\{v\}$. Members of $\mathcal{I}$ are called independent sets. A base of $\mathcal{M}$ is a maximal independent set (\emph{w.r.t.} inclusion). All bases of $\mathcal{M}$ have equal cardinality, this common cardinality is called the rank of $\mathcal{M}$, denoted $r_{\mathcal{M}}(N)$. Throughout the paper, we adopt the computational viewpoint that a matroid $\mathcal{M}=(N,\mathcal{I})$ is presented to an algorithm solely via its ground set $N$ and an independence oracle answering whether any given subset $I\subseteq N$ is independent; accordingly, all matroids under discussion are assumed to be equipped with such an oracle.

The matroid base packing problem asks for a maximum cardinality collection of
pairwise disjoint bases of \(\mathcal M\). Its optimal value is called the packing
number and is denoted by \(\rho(\mathcal M)=\rho^*\). Throughout the paper,
unless otherwise stated, a base packing means a maximum base packing.

\begin{definition}
Let $\mathcal{M}=(N,\mathcal{I})$ be a matroid . The associated matroid packing game $\Gamma_{\mathcal{M}}=(N,\gamma)$ is defined as
\begin{enumerate}
\item[(1)] The player set is grand set $N$;
\item[(2)] $\forall\,S\subseteq N$, $\gamma(S)$ is the maximum number of pairwise disjoint bases of \(\mathcal M\) contained in \(S\).
\end{enumerate}
\end{definition}

Throughout the paper, we assume that $r_{\mathcal M}(N)>0$ and that $N$ is not a base of $\mathcal M$. If $N$ is a base, then $\gamma(N)=1$ and $\gamma(S)=0$ for every proper subset $S$ of $N$. In this case, the nucleolus assigns $1/|N|$ to every player.

\begin{remark}
When \(\mathcal M\) is the graphic matroid of a connected graph \(G=(V,E)\),
the bases of \(\mathcal M\) are the spanning trees of \(G\). Hence the associated
matroid packing game specializes to the network strength game\cite{BB20}, where the players
are the edges and the value of an edge set \(S\subseteq E\) is the maximum number
of pairwise edge disjoint spanning trees contained in \(S\).
\end{remark}

\subsection{Matroid base packing}\label{sec:packing-preliminaries}

Let $\mathcal B$ be the collection of bases of $\mathcal M$. For $S\subseteq N$, let $r_{\mathcal M}(S)$ denote the maximum size of an independent subset of $S$. The base packing problem can be written as
\begin{equation}\label{eq:integer-packing}
(\mathrm{IP}):\quad
\begin{array}{rl}
\mbox{max}&\displaystyle\sum_{B\in\mathcal B}\lambda_B\\[1mm]
\mbox{s.t.}&\left\{
\begin{array}{ll}
\displaystyle\sum_{B\in\mathcal B:v\in B}\lambda_B\leq1,
&\quad\forall v\in N,\\[2mm]
\lambda_B\geq0,\ \text{integer valued},&\quad\forall B\in\mathcal B.
\end{array}\right.
\end{array}
\end{equation}
Each feasible $\lambda_B$ is either zero or one, and the bases with $\lambda_B=1$ are pairwise disjoint. Thus the optimal value of $(\mathrm{IP})$ is $\rho^*$. A maximum base packing can be computed in polynomial time~\cite[Theorem~42.5]{Schr03}.

Relaxing the integrality condition gives the fractional base packing problem. Its optimal value is denoted by $\rho_f^*$. The linear program and its dual are
\begin{equation}\label{eq:fractional-packing-dual}
\begin{array}{rll}
(\mathrm{FP}):&\displaystyle\max&\displaystyle\sum_{B\in\mathcal B}\lambda_B\\[1mm]
&\mbox{s.t.}&\displaystyle\sum_{B\in\mathcal B:v\in B}\lambda_B\leq1,
\quad\forall v\in N,\\[1mm]
&&\lambda_B\geq0,\quad\forall B\in\mathcal B,
\end{array}
\qquad
\begin{array}{rll}
(\mathrm{FD}):&\min&y(N)\\[1mm]
&\mbox{s.t.}&y(B)\geq1,\quad\forall B\in\mathcal B,\\[1mm]
&&y(v)\geq0,\quad\forall v\in N.
\end{array}
\end{equation}
By linear programming duality, both programs have optimal value $\rho_f^*$. A maximum fractional base packing can be computed in strongly polynomial time, represented by a polynomial number of bases and their positive rational weights~\cite[Corollary~42.7a]{Schr03}. These formulations will be used in Section~5 to characterize the least core.

\subsection{Graph characterization of base packing and fractional base packing}\label{sec:packing-exchange-graphs}

Let $\mathcal{M}=(N,\mathcal{I})$ be a matroid, and let $\mathcal{B}$ be the collection of bases of $\mathcal{M}$. In this subsection, we introduce the base exchange graphs associated with bases, base packings and fractional base packings, using the exchange structure of matroids~\cite{Schr03}.

\begin{definition}\label{base-exchange-graphs}\cite{Schr03}
Let $\mathcal{M}=(N,\mathcal{I})$ be a matroid.
\begin{enumerate}
\item[(1)] For $B\in\mathcal{B}$, define the (bipartite) directed base exchange graph $G_{B}=(N,E_{B})$ by:
\begin{enumerate}
\item[a)] The vertex set is $N$;
\item[b)] The (directed) edge set is $E_{B}=\{(u,v)|\,u\in B,v\in N\setminus B,\,B-u+v\in\mathcal{B}\}$.
\end{enumerate}
\item[(2)] For a base packing $\mathcal{P}_{B}=\{B_{1},...,B_{\rho^{*}}\}$, define the directed  graph $G[\mathcal{P}_{B}]=(N,E[\mathcal{P}_{B}])$ by the union of the graphs $G_{B_{i}}=(N,E_{B_{i}})$ $(i=1,2,...,\rho^{*})$. That is,
\begin{enumerate}
\item[a)] The vertex set is $N$;
\item[b)] The (directed) edge set is $E[\mathcal{P}_{B}]=\bigcup_{i=1}^{\rho^{*}}E_{B_{i}}$.
\end{enumerate}
\item[(3)] For a fractional base packing with positive weight bases $\widetilde{\mathcal P}_B=\{\widetilde B_1,\ldots,\widetilde B_s\}$ and weights $\lambda_1,\ldots,\lambda_s>0$, define the directed graph $G[\widetilde{\mathcal P}_B]=(N,E[\widetilde{\mathcal P}_B])$ by the union of the graphs $G_{\widetilde B_i}$ $(i=1,\ldots,s)$. That is,
\begin{enumerate}
\item[a)] The vertex set is $N$;
\item[b)] The (directed) edge set is $E[\widetilde{\mathcal P}_B]=\bigcup_{i=1}^s E_{\widetilde B_i}$.
\end{enumerate}
The graph is determined by the positive weight bases, which need not be pairwise disjoint.
\end{enumerate}
\end{definition}



\begin{lemma}\label{matroid exchange 1}(Corollary 39.12a in \cite{Schr03})
Let $\mathcal{M}=(N,\mathcal{I})$ be a matroid, and let $B_{1},B_{2}\in\mathcal{B}$. Then $E_{B_{1}}$ contains a perfect matching between $B_{1}\setminus B_{2}$ and $B_{2}\setminus B_{1}$.
\end{lemma}

For a base packing $\mathcal P_B$, let $N[\mathcal P_B]=\bigcup_{B\in\mathcal P_B}B$ denote its support.

Let $\mathcal P_B=\{B_1,\ldots,B_s\}$ be the positive weight bases of a fractional base packing, with weights $\lambda_1,\ldots,\lambda_s>0$, and let $G[\mathcal P_B]$ be its base exchange graph. For each $v\in N$, let $c(v)=\sum_{i:v\in B_i}\lambda_i$ be its load, and let $N_0=\{v\in N:c(v)<1\}$ be the set of elements whose capacity constraints are not tight.

Let $\mathcal Q$ be the collection of all strongly connected components (SCCs) of $G[\mathcal P_B]$, and let
\[
\mathcal Q^+=\{Q\in\mathcal Q:\text{there is no directed path from }Q\text{ to }N_0\}.
\]
For convenience, a $v$-$U$ path means a directed path from $v$ to some vertex in $U$, and a $U$-$U'$ path means a directed path from some vertex in $U$ to some vertex in $U'$. Paths of length zero are allowed. In particular, every member of $\mathcal Q^+$ is disjoint from $N_0$. When no confusion arises, we identify an SCC with its vertex set. Given the fractional base packing and its exchange graph, $N_0$, $\mathcal Q$ and $\mathcal Q^+$ can be computed in polynomial time.

The SCCs of $G[\mathcal P_B]$ form a directed acyclic condensation graph. For $Q,Q'\in\mathcal Q$, define the partial order $\preceq$ by
\[
Q\preceq Q'\quad\Longleftrightarrow\quad
\text{there is a directed path from }Q\text{ to }Q'\text{ in }G[\mathcal P_B].
\]
We use the same notation for its restriction to $\mathcal Q^+$.

\begin{remark}
A base packing $\mathcal P_B=\{B_1,\ldots,B_{\rho^*}\}$ is a fractional base packing with $\lambda_i=1$ for every $i$. In this case, $G[\mathcal P_B]$ is the base exchange graph of the integer packing, $c(v)=1$ on $N[\mathcal P_B]$ and $c(v)=0$ elsewhere. Thus $N_0=N\setminus N[\mathcal P_B]$ is precisely the set of elements not covered by the chosen maximum family of pairwise disjoint bases. Every $v\in N_0$ has out degree zero and forms a singleton SCC. These singleton components belong to $\mathcal Q$ but not to $\mathcal Q^+$.
\end{remark}

%
%

The method of defining a partially ordered set and computing the nucleolus based on its hierarchy partition is described in \cite{ALPS09}, \cite{BB20} and \cite{XF23}. This approach remains applicable to the computation of the nucleolus in matroid packing games. Therefore, we define the hierarchy partition of the partially ordered set and provide the corresponding algorithm.

When $\mathcal Q^+\ne\emptyset$, the hierarchy partition of $(\mathcal Q^+,\preceq)$ is denoted by $\mathcal Q_1^+\prec\cdots\prec\mathcal Q_t^+$. It is computed iteratively: at each step, all minimal elements of the remaining poset are removed, and the elements removed in step $i$ form $\mathcal Q_i^+$. Thus the hierarchy partition satisfies the following properties:
\begin{enumerate}
\item[(i)] $\mathcal Q_1^+$ consists of all minimal elements of $(\mathcal Q^+,\preceq)$;
\item[(ii)] For every $Q\in\mathcal Q_i^+$ with $i>1$, there exists $Q'\in\mathcal Q_{i-1}^+$ such that $Q'\preceq Q$;
\item[(iii)] For distinct $Q,Q'\in\mathcal Q_i^+$, neither $Q\preceq Q'$ nor $Q'\preceq Q$ holds.
\end{enumerate}

\begin{lemma}\label{SSC2}
The hierarchy partition of $(\mathcal Q^+,\preceq)$ can be computed in polynomial time.
\end{lemma}

\begin{center}

\begin{figure}
	 \centering
  \includegraphics[width=0.65\linewidth]{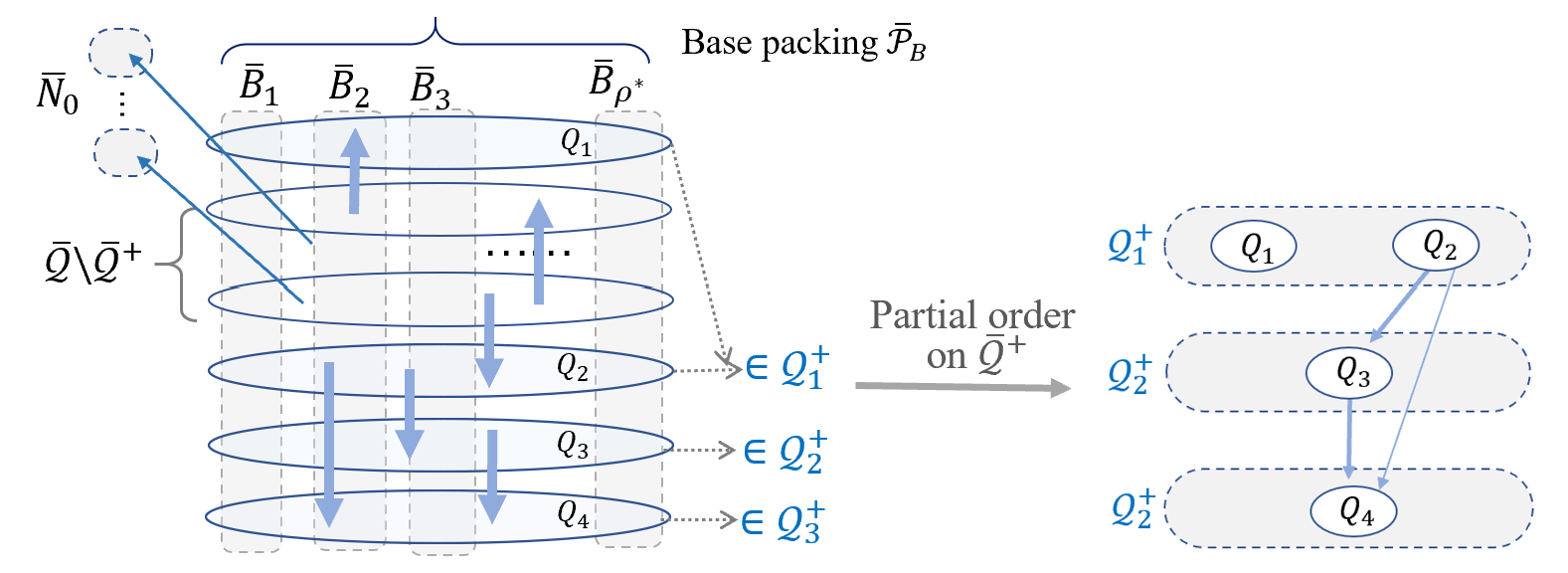}
	\caption{The hierarchy partition}
	\label{fig:one}
\end{figure}
\end{center}

\section{Characterization of the core of $\Gamma_{\mathcal{M}}$}

In this section, we characterize the core of the matroid packing game by using a fixed base packing and its exchange graph. Let \(\mathcal M=(N,\mathcal I)\) be a matroid and let \(\Gamma_{\mathcal M}=(N,\gamma)\) be the associated matroid packing game. The set of allocations of \(\Gamma_{\mathcal M}\) is \[ X(\Gamma_{\mathcal M}) = \{x\in\mathbb R_{\ge0}^{N}:x(N)=\rho^*\}. \]

\subsection{The essential coalitions}

Let $\Gamma=(N,\gamma)$ be a cooperative game. A subset $S\subseteq N$ is called an essential coalition of $\Gamma$ if either $|S|=1$, or
$
\gamma(S)>\sum_{T\in\mathcal{T}}\gamma(T)
$
for every nontrivial partition $\mathcal{T}$ of $S$, where a partition of $S$ is called trivial if it consists only of the coalition $S$ itself. We denote by $\mathcal{E}(\Gamma)$ the set of all essential coalitions of $\Gamma$. It was shown in \cite{Gur} that for the game $\Gamma=(N,\gamma)$:
\begin{enumerate}
\item[(1)] The core of $\Gamma$ can be determined only by essential coalitions;
\item[(2)] When the core is nonempty, dropping the constraints associated with inessential coalitions does not change the result of $SLP(\eta(\Gamma))$ for computing the nucleolus. That is, the nucleolus can be determined completely by essential coalitions.
\end{enumerate}

\begin{lemma}\label{core 1}
The set of essential coalitions of the matroid packing game $\Gamma_{\mathcal{M}}$ consists of all singletons and all coalitions corresponding to bases of $\mathcal{M}$. That is,
\[
\mathcal{E}(\Gamma_{\mathcal{M}})
=
\{\{v\}:\,v\in N\}\cup\{B:B\in\mathcal{B}\}.
\]
\end{lemma}

Following directly from Lemma \ref{core 1}, we have the following characterization.

\begin{theorem}\label{core 1.1}
Let $\Gamma_{\mathcal{M}}=(N,\gamma)$ be the matroid packing game. Then
\[
\mathcal{C}(\Gamma_{\mathcal{M}})
=
\{x\in X(\Gamma_{\mathcal{M}}):x(B)\geq 1,\ \forall B\in\mathcal{B}\}.
\]
Furthermore, for each $x\in\mathcal{C}(\Gamma_{\mathcal{M}})$ and each base packing $\mathcal{P}_{B}=\{B_{1},...,B_{\rho^{*}}\}$, it holds that:
\begin{enumerate}
\item[(a)] $x(B_{i})=1$ for all $i\in\{1,2,...,\rho^{*}\}$, and
\item[(b)] $x(v)=0$ for all $v\in N\setminus N[\mathcal P_B]$.
\end{enumerate}
\end{theorem}

\begin{proof}
By Lemma \ref{core 1} and the nonnegativity constraints in $X(\Gamma_{\mathcal M})$, the core is determined by the constraints corresponding to bases. Indeed, singleton constraints are either nonnegativity constraints, or are already included among the base constraints when a singleton is a base. Hence
\[
\mathcal{C}(\Gamma_{\mathcal{M}})
=
\{x\in X(\Gamma_{\mathcal{M}}):x(B)\geq 1,\ \forall B\in\mathcal{B}\}.
\]

Let $x\in\mathcal{C}(\Gamma_{\mathcal{M}})$ and let $\mathcal{P}_{B}=\{B_{1},...,B_{\rho^{*}}\}$ be a base packing. Since the bases in $\mathcal P_B$ are pairwise disjoint, we have
\[
x(N)
=
\sum_{i=1}^{\rho^{*}}x(B_{i})
+
\sum_{v\in N\setminus N[\mathcal P_B]}x(v)
\geq
\rho^{*}
=
\gamma(N).
\]
Since $x(N)=\gamma(N)=\rho^*$,  $x(B_i)\geq 1$ for  $i=1,\dots, \rho^*$ and $x(v)\geq 0$ for  $v\in N$, we have
$
x(B_i)=1,\quad i=1,\ldots,\rho^*,
$
and
$
x(v)=0,\quad \forall v\in N\setminus N[\mathcal P_B].
$
\end{proof}

\subsection{The characterizations  and algorithms for the core}

Now we employ a fixed base packing to establish a more refined characterization of the core $\mathcal{C}(\Gamma_{\mathcal{M}})$. Let
\[
\bar{\mathcal{P}}_{B}=\{\bar{B}_{1},\bar{B}_{2},...,\bar{B}_{\rho^{*}}\}
\]
be a fixed base packing, let $G[\bar{\mathcal{P}}_{B}]$ be the related exchange graph,
$
N[\bar{\mathcal{P}}_{B}]
=
\bigcup^{\rho^{*}}_{i=1}\bar{B}_{i}
$ and
$\bar{N}_{0}=N\setminus N[\bar{\mathcal{P}}_{B}].
$
Let $\bar{\mathcal{Q}}$, $\bar{\mathcal{Q}}^{+}$ and the  partial order ``$\preceq$'' be defined as in Section 2 with respect to the graph $G[\bar{\mathcal{P}}_{B}]$. Concretely, for SCCs $Q,Q^{\prime}\in\bar{\mathcal{Q}}^{+}$, we write $Q\preceq Q^{\prime}$ if and only if there is a directed $Q$-$Q^{\prime}$ path  in $G[\bar{\mathcal{P}}_{B}]$. A graph characterization of $\mathcal{C}(\Gamma_{\mathcal{M}})$ is established by Theorem \ref{core 2}.

\begin{theorem}\label{core 2}
Let $x\in \mathbb{R}^{N}_{\geq 0}$ and let $\bar{\mathcal{P}}_{B}=\{\bar{B}_{1},\bar{B}_{2},...,\bar{B}_{\rho^{*}}\}$ be a fixed base packing. Then $x\in\mathcal{C}(\Gamma_{\mathcal{M}})$ if and only if
\begin{enumerate}
\item[(a)] $x(N)=\gamma(N)$;
\item[(b)] $x(v)=0$ for $v\in\bar{N}_{0}$;
\item[(c)] $x(u)\leq x(v)$ for $(u,v)\in E[\bar{\mathcal{P}}_{B}]$.
\end{enumerate}
\end{theorem}

\begin{proof}
Let $x\in\mathcal C(\Gamma_{\mathcal M})$. Condition (a) follows from $x\in X(\Gamma_{\mathcal M})$, and (b) follows from Theorem~\ref{core 1.1}. To prove (c), let $(u,v)\in E[\bar{\mathcal P}_B]$. There exists $\bar B_i\in\bar{\mathcal P}_B$ such that $B'=\bar B_i-u+v\in\mathcal B$. By Theorem~\ref{core 1.1}, $x(\bar B_i)=1$ and $x(B')\geq1$. Thus $x(v)-x(u)=x(B')-x(\bar B_i)\geq0$.

Conversely, let $x\in\mathbb R_{\geq0}^N$ satisfy (a), (b) and (c). For every $i\in\{1,\ldots,\rho^*\}$ and every $B\in\mathcal B$, Lemma~\ref{matroid exchange 1} gives a perfect matching from $\bar B_i\setminus B$ to $B\setminus\bar B_i$ in $E_{\bar B_i}$. Summing (c) over this matching yields $x(B)\geq x(\bar B_i)$. Hence all bases in $\bar{\mathcal P}_B$ are minimum weight bases and have the same weight. Since they are pairwise disjoint, (a) and (b) give
\[
\rho^*=\gamma(N)=x(N)=\sum_{i=1}^{\rho^*}x(\bar B_i)=\rho^*x(\bar B_1).
\]
Therefore $x(\bar B_i)=1$ for every $i$, and $x(B)\geq1$ for every $B\in\mathcal B$. By (a), $x\in X(\Gamma_{\mathcal M})$, so Theorem~\ref{core 1.1} gives $x\in\mathcal C(\Gamma_{\mathcal M})$.
\end{proof}

\begin{corollary}\label{core 2.1}
Given a base packing $\bar{\mathcal{P}}_{B}=\{\bar{B}_{1},\bar{B}_{2},...,\bar{B}_{\rho^{*}}\}$ and its related exchange graph $G[\bar{\mathcal{P}}_{B}]$, if $x\in\mathcal{C}(\Gamma_{\mathcal{M}})$, then
\begin{enumerate}
\item[(a)] If $u\in Q$, $v\in Q^{\prime}$ and $Q\preceq Q^{\prime}$, then $x(u)\leq x(v)$;
\item[(b)] If $u,v$ belong to the same SCC $Q\in\bar{\mathcal{Q}}$, then $x(u)=x(v)$. Henceforth we write $x_{Q}$ to denote the common value of the vertices in $Q$, i.e.,
\[
x_{Q}=\frac{x(Q)}{|Q|};
\]
\item[(c)] If there is a directed path from $v$ to $\bar{N}_{0}$ in $G[\bar{\mathcal{P}}_{B}]$, then $x(v)=0$. Equivalently,
\[
x(v)=0,\quad \forall v\in Q,\ Q\in\bar{\mathcal{Q}}\setminus\bar{\mathcal{Q}}^{+}.
\]
\end{enumerate}
\end{corollary}

Theorem~\ref{core 2} and Corollary~\ref{core 2.1} suggest a natural way to construct a core allocation when $\bar{\mathcal Q}^{+}\neq\emptyset$. Let $q=|\bigcup_{Q\in\bar{\mathcal Q}^{+}}Q|$, assign the value $\rho^*/q$ to every vertex in this union, and assign zero to all remaining vertices. The resulting vector is nonnegative and satisfies $x(N)=\gamma(N)$. Since no directed edge leaves this union, conditions (b) and (c) of Theorem~\ref{core 2} also hold.
A concise necessary and sufficient condition for the nonemptiness of the core is obtained as follows.

\begin{theorem}\label{core 3}
For a fixed base packing $\bar{\mathcal P}_B$, the core $\mathcal C(\Gamma_{\mathcal M})$ is nonempty if and only if
\[
\bar{\mathcal{Q}}^{+}\neq\emptyset.
\]
\end{theorem}

Combining Theorems~\ref{core 2} and~\ref{core 3}, we obtain the following result.

\begin{theorem}
Let $\Gamma_{\mathcal{M}}=(N,\gamma)$ be the matroid packing game defined on the matroid $\mathcal{M}=(N,\mathcal{I})$. Then, for the core of $\Gamma_{\mathcal{M}}$, the problems of testing nonemptiness, checking membership, and finding a core member can all be solved in polynomial time.
\end{theorem}

\section{The nucleolus when $\mathcal{C}(\Gamma_{\mathcal{M}})\neq\emptyset$}

This section is devoted to an efficient algorithm for computing the nucleolus of a matroid packing game with nonempty core. Our methodology relies on the exchange graph characterization of the core established in the previous section. We show that the nucleolus $\eta(\Gamma_{\mathcal{M}})$ is determined by the first two linear programs $(\textnormal{LP})_1$ and $(\textnormal{LP})_2$ in $SLP(\eta(\Gamma_{\mathcal M}))$ defined in \eqref{nucleolus LP}.

For $(\textnormal{LP})_1$, we have $\varepsilon_1=0$ and
$
X_1=\mathcal{C}(\Gamma_{\mathcal{M}}).
$
Indeed, since the core is nonempty, $\varepsilon\geq 0$  and $X_{1}\subseteq \mathcal{C}(\Gamma_{\mathcal{M}})$. On the other hand, for any base packing $\mathcal P_B=\{B_1,\ldots,B_{\rho^*}\}$ and $x\in \mathcal{C}(\Gamma_{\mathcal{M}})$, the condition  $x(B_i)= 1$  implies
that $\varepsilon=0$ and $X_{1}=\mathcal{C}(\Gamma_{\mathcal{M}})$.

As established in \cite{Gur}, when the core is nonempty, the nucleolus is fully determined by the essential coalitions. For the matroid packing game $\Gamma_{\mathcal{M}}$, Lemma \ref{core 1} shows that the essential coalitions are exactly the singletons and the bases. Consequently, $(\textnormal{LP})_2$ in $\eta(\Gamma_{\mathcal{M}})$ can be written as follows:
\[
{\text{(LP)}}_2: \quad
\begin{array}{rl}
\mbox{max} & \varepsilon \\[1mm]
\mbox{s.t.} &
\left\{
\begin{array}{ll}
x(v)\geq \varepsilon
& \quad \forall v\in N\setminus\displaystyle N_1, \\[2mm]
x(B)\geq 1+\varepsilon
& \quad \forall B\in {\mathcal B}\setminus\displaystyle{\mathcal B}_1, \\[2mm]
x\in \mathcal{C}(\Gamma_{\mathcal{M}}).
\end{array}
\right.
\end{array}
\]
Here
$
N_1=\{v\in N:x(v)=0,\ \forall x\in \mathcal{C}(\Gamma_{\mathcal{M}})\},
$
and
$
{\mathcal B}_1=\{B\in\mathcal B:x(B)=1,\ \forall x\in \mathcal{C}(\Gamma_{\mathcal{M}})\}.
$

Let
$
\bar{\mathcal{P}}_{B}=\{\bar{B}_{1},\bar{B}_{2},\ldots,\bar{B}_{\rho^{*}}\}
$
be a fixed base packing. For an arbitrary $B\in\mathcal{B}$, by Lemma \ref{matroid exchange 1}, there is a perfect matching
$
\pi=\{(u_{1},v_{1}),\ldots,(u_{t},v_{t})\}
$ from $\bar{B}_{1}\setminus B$ to $B\setminus\bar{B}_{1}$. Hence, for  $x\in \mathcal{C}(\Gamma_{\mathcal{M}})$,
\begin{equation}\label{eq:excess.base}
x(B)
=
x({\bar B}_1)+\sum_{k=1}^t\left[x(v_k)-x(u_k)\right]
=
1+\sum_{k=1}^t\left[x(v_k)-x(u_k)\right],
\end{equation}
where $x(\bar{B}_{1})=1$ by Theorem \ref{core 1.1}. Thus, $x(B)$ is determined by the local differences $x(v)-x(u)$ along arcs of $E[\bar{\mathcal{P}}_{B}]$.

This leads to the following polynomial size linear program:
\[
(\textnormal{LP})'_2:\quad
\begin{array}{rl}
\mbox{max} & \varepsilon \\[1mm]
\mbox{s.t.} &
\left\{
\begin{array}{ll}
x(v)=0
& \quad \forall v\in {N}_{*}, \\[1mm]
x(v)\geq \varepsilon
& \quad \forall v\in N\setminus {N}_{*}, \\[1mm]
x(v)-x(u)=0
& \quad \forall (u,v)\in {E}_{*}, \\[1mm]
x(v)-x(u)\geq \varepsilon
& \quad \forall (u,v)\in E[\bar{\mathcal{P}}_{B}]\setminus {E}_{*}, \\[1mm]
x(N)=\gamma(N).
\end{array}
\right.
\end{array}
\]
Here
$
N_{*}
=
\{v\in N:x(v)=0,\ \forall x\in\mathcal{C}(\Gamma_{\mathcal{M}})\}=N_1
$
and
$
E_{*}
=
\{(u,v)\in E[\bar{\mathcal{P}}_{B}]:x(v)-x(u)=0,\ \forall x\in\mathcal{C}(\Gamma_{\mathcal{M}})\}.
$

\begin{lemma}\label{noncore nucleolus eq 1}
$(\textnormal{LP})'_2$ is equivalent to $(\textnormal{LP})_2$ in $SLP(\eta(\Gamma_{\mathcal M}))$.
\end{lemma}

\begin{proof}
Every core allocation is feasible for $(\mathrm{LP})'_2$ at $\varepsilon=0$, and summing the vertex constraints gives $|N\setminus N_*|\varepsilon\leq\gamma(N)$. Since $N\setminus N_*\ne\emptyset$, the program has a finite optimum. Let $(x,\varepsilon)$ be an optimal solution. Then $\varepsilon\geq0$, so its vertex and edge constraints imply $x\geq0$ and $x(u)\leq x(v)$ for every $(u,v)\in E[\bar{\mathcal P}_B]$. Moreover, $\bar N_0\subseteq N_*$ by Theorem~\ref{core 1.1}, so $x(v)=0$ for $v\in\bar N_0$. Together with $x(N)=\gamma(N)$, Theorem~\ref{core 2} yields $x\in\mathcal C(\Gamma_{\mathcal M})$. In particular, $x(\bar B_i)=1$ for every $i$ by Theorem~\ref{core 1.1}.

The vertex constraints in the two programs are the same, since $N_*=N_1$. We now compare the base constraints. Let $B\in\mathcal B\setminus\mathcal B_1$, and choose a perfect matching $\pi=\{(u_1,v_1),\ldots,(u_t,v_t)\}$ from $\bar B_1\setminus B$ to $B\setminus\bar B_1$. By~\eqref{eq:excess.base},
\[
x(B)=1+\sum_{k=1}^t[x(v_k)-x(u_k)].
\]
If every arc of $\pi$ belonged to $E_*$, this identity would give $x'(B)=1$ for every $x'\in\mathcal C(\Gamma_{\mathcal M})$, contrary to $B\notin\mathcal B_1$. Hence at least one arc lies outside $E_*$. Its difference is at least $\varepsilon$, and every other difference is nonnegative. Thus $x(B)\geq1+\varepsilon$, proving feasibility for $(\mathrm{LP})_2$.

Conversely, let $(x,\varepsilon)$ be feasible for $(\mathrm{LP})_2$. Then $x\in\mathcal C(\Gamma_{\mathcal M})$, so $x(N)=\gamma(N)$ and the zero coordinate and zero difference constraints of $(\mathrm{LP})'_2$ hold. Its vertex inequalities also hold since $N_*=N_1$. Let $(u,v)\in E[\bar{\mathcal P}_B]\setminus E_*$. There exists $x'\in\mathcal C(\Gamma_{\mathcal M})$ such that $x'(u)<x'(v)$. Choose $\bar B_i$ with $B'=\bar B_i-u+v\in\mathcal B$. By Theorem~\ref{core 1.1}, $x'(B')=1-x'(u)+x'(v)>1$, so $B'\notin\mathcal B_1$. Its base constraint gives $x(B')\geq1+\varepsilon$. Since $x(\bar B_i)=1$, we obtain $x(v)-x(u)=x(B')-x(\bar B_i)\geq\varepsilon$. Therefore, the two programs have the same optimal value and the same optimal allocations.
\end{proof}

We next determine the sets $N_{*}$ and $E_{*}$ with respect to the core. Recall the  partial order ``$\preceq$'' defined on the SCCs of $G[\bar{\mathcal{P}}_{B}]$. Consider the hierarchy partition
$
\mathcal{Q}^{+}_{1}\prec...\prec\mathcal{Q}^{+}_{t}
$
of  $(\bar{\mathcal{Q}}^{+},\preceq)$ introduced in Section 2.

For $i=1,\ldots,t$, set
$
n_i=|\bigcup_{Q\in\bar{\mathcal Q}_i^+}Q|
$
and $
q=\sum_{i=1}^{t}n_i,
$
which is the number of elements contained in $\bar{\mathcal Q}_{i}^{+}$ and  $\bar{\mathcal Q}^{+}$.

We shall use a core allocation  that is strictly increasing along the hierarchy layers. For each \(Q\in\bar{\mathcal Q}_j^+\), define

\begin{equation}\label{core constru}
\bar{x}_{Q}=\alpha+(j-1)\beta,\qquad j=1,\ldots,t,
\end{equation}
where, if $t\geq2$, we choose $0<\alpha<\rho^*/q$ and $\beta=(\rho^*-q\alpha)/\sum_{j=1}^{t}(j-1)n_j$. If $t=1$, take $\alpha=\rho^*/q$ and $\beta=0$.

Set $\bar x(v)=0$ for $v\in N\setminus\bigcup_{Q\in\bar{\mathcal Q}^+}Q$. Then $\bar x(N)=\gamma(N)$, and $\bar x$ satisfies the zero coordinate and edge inequalities in Theorem~\ref{core 2}. Thus $\bar x\in\mathcal C(\Gamma_{\mathcal M})$. Moreover, \(\bar{x}_{Q}>0\) for each \(Q\in\bar{\mathcal{Q}}^{+}\), and, for \(Q\prec Q'\) with \(Q\neq Q'\), we have
$
\bar x_Q<\bar x_{Q'}.
$

\begin{lemma}\label{NEchar}
If $\mathcal{C}(\Gamma_{\mathcal{M}})\neq\emptyset$, then
\begin{enumerate}
\item[(a)] $v\in N_{*}$ if and only if
$
v\in  N\setminus\bigcup_{Q\in\bar{\mathcal Q}^+}Q;
$
\item[(b)] For $(u,v)\in E[\bar{\mathcal P}_B]$, $(u,v)\in E_*$ if and only if $u,v\in N_*$ or $u,v$ belong to the same member of $\bar{\mathcal Q}^+$.
\end{enumerate}
\end{lemma}

\begin{lemma}\label{nucleolus2}
The unique optimal allocation of $(\mathrm{LP}'_2)$ is given by
\begin{equation}\label{nucleolus chater}
x_Q=i\varepsilon_2^*
\quad(Q\in\bar{\mathcal Q}_i^+),\qquad
\varepsilon_2^*=\frac{\rho^*}{n_1+2n_2+\cdots+tn_t},
\end{equation}
and $x(v)=0$ for $v\in N_*$.
\end{lemma}

\begin{proof}
Let $(x,\varepsilon)$ be an optimal solution of $(\mathrm{LP}'_2)$. Every core allocation is feasible for this program at $\varepsilon=0$, so $\varepsilon\geq0$. By Theorem~\ref{core 2}, $x\in\mathcal C(\Gamma_{\mathcal M})$, and $x$ is constant on each SCC by Corollary~\ref{core 2.1}. For $Q\in\bar{\mathcal Q}_1^+$, the vertex constraints give $x_Q\geq\varepsilon$. For $Q\in\bar{\mathcal Q}_i^+$ with $i>1$, the hierarchy partition gives a predecessor $Q'\in\bar{\mathcal Q}_{i-1}^+$ joined to $Q$ by an exchange edge. By Lemma~\ref{NEchar}, this edge does not belong to $E_*$, so $x_Q\geq x_{Q'}+\varepsilon$. Induction gives $x_Q\geq i\varepsilon$. The total constraint therefore implies
\[
\rho^*=\gamma(N)=x(N)=\sum_{i=1}^t\sum_{Q\in\bar{\mathcal Q}_i^+}|Q|x_Q
\geq\varepsilon(n_1+2n_2+\cdots+tn_t).
\]
Thus $\varepsilon\leq\rho^*/(n_1+2n_2+\cdots+tn_t)$. At this upper bound, assign $x_Q=i\varepsilon$ for $Q\in\bar{\mathcal Q}_i^+$ and zero to $N_*$. The resulting vector has total $\gamma(N)$ and satisfies all vertex and edge constraints of $(\mathrm{LP}'_2)$, since every edge between distinct members of $\bar{\mathcal Q}^+$ goes from an earlier layer to a later layer. Hence the bound is attained. Equality in the displayed sum forces $x_Q=i\varepsilon$ on every block, proving uniqueness.
\end{proof}

Lemma~\ref{nucleolus2} explicitly gives the unique optimal solution of
\((\textnormal{LP})'_2\), which is the nucleolus \(\eta(\Gamma_{\mathcal M})\).
Since the hierarchy partition can be computed in polynomial time,  we have the follows result.
\begin{theorem}
Let $\Gamma_{\mathcal{M}}=(N,\gamma)$ be the matroid packing game defined on the matroid $\mathcal{M}=(N,\mathcal{I})$. When $\mathcal{C}(\Gamma_{\mathcal{M}})\neq\emptyset$, the nucleolus $\eta(\Gamma_{\mathcal{M}})$ can be computed in polynomial time.
\end{theorem}

\section{The least core and nucleolus when $\mathcal{C}(\Gamma_{\mathcal{M}}) = \emptyset$}

In this section, we discuss the least core and the nucleolus of the matroid packing game $\Gamma_{\mathcal{M}}$ with empty core. We first characterize the least core by the exchange graph of a fixed maximum fractional base packing. The SCCs of this graph determine a partition and a partial order. We then show that, as in the nonempty core case, the hierarchy partition determines the nucleolus through the second linear program in SLP$(\eta(\Gamma_{\mathcal M}))$.

\subsection{The least core}\label{sec:least-core}

We first consider the least core $\mathcal{LC}(\Gamma_{\mathcal M})$, which is the optimal solution set of the first linear program $(\mathrm{LP})_1$ in SLP$(\eta(\Gamma_{\mathcal M}))$:
\[
{\textnormal{(LP)}}_1:\quad
\begin{array}{rl}
\mbox{max} & \varepsilon\\[1mm]
\mbox{s.t.}
& \left\{
\begin{array}{ll}
x(S)\geq \gamma(S)+\varepsilon, & \quad \forall S\subseteq N,\\[1mm]
x(N)=\rho^*, & \\[1mm]
x(v)\geq 0, & \quad \forall v\in N.
\end{array}
\right.
\end{array}
\]
Let $\varepsilon_1^*$ be its optimal value, and write $X_1=\mathcal{LC}(\Gamma_{\mathcal M})$.

Recall the fractional packing program $(\mathrm{FP})$, its dual $(\mathrm{FD})$ and their common optimal value $\rho_f^*$ from Section~\ref{sec:packing-preliminaries}. Fix a maximum fractional base packing with positive weight bases $\widetilde{\mathcal P}_B=\{\widetilde B_1,\ldots,\widetilde B_s\}$ and weights $\lambda_1,\ldots,\lambda_s>0$. Then $\sum_{i=1}^s\lambda_i=\rho_f^*$. Let $G[\widetilde{\mathcal P}_B]=(N,E[\widetilde{\mathcal P}_B])$ be its base exchange graph, where $E[\widetilde{\mathcal P}_B]=\bigcup_{i=1}^sE_{\widetilde B_i}$. Define the load vector $c$ by $c(v)=\sum_{i:v\in\widetilde B_i}\lambda_i$ for $v\in N$, and let $\widetilde N_0=\{v\in N:c(v)<1\}$. By feasibility of $(\mathrm{FP})$, $0\leq c(v)\leq1$ for every $v\in N$. This packing, its graph and the associated notation remain fixed throughout this section.

Since an integer base packing is also a fractional base packing, $\rho_f^*\geq\rho^*$. If equality holds, an optimal solution of $(\mathrm{FD})$ is a core allocation. Consequently, the assumption $\mathcal C(\Gamma_{\mathcal M})=\emptyset$ implies $\rho_f^*>\rho^*$.

\begin{lemma}\label{fractional-load-decomposition}
The vector $c/\rho_f^*$ is a convex combination of the incidence vectors of bases of $\mathcal M$, and $(\rho^*/\rho_f^*)c$ is a convex combination of the incidence vectors of sets in $\mathcal B^{\rho^*}$, where $\mathcal B^{\rho^*}$ is the collection of unions of $\rho^*$ pairwise disjoint bases of $\mathcal M$.
\end{lemma}

\begin{proof}
For $S\subseteq N$, let $\chi^S$ denote its incidence vector. Since $\lambda_i/\rho_f^*>0$ and $\sum_{i=1}^s\lambda_i/\rho_f^*=1$, the identity $c/\rho_f^*=\sum_{i=1}^s(\lambda_i/\rho_f^*)\chi^{\widetilde B_i}$ proves the first assertion.

For the second assertion, consider the union matroid $\mathcal M^{\rho^*}=(N,\mathcal I^{\rho^*})$ of $\rho^*$ copies of $\mathcal M$, where $\mathcal I^{\rho^*}=\{I_1\cup\cdots\cup I_{\rho^*}:I_j\in\mathcal I,\ j=1,\ldots,\rho^*\}$. By the matroid union theorem~\cite[Corollary~42.1a]{Schr03}, this is a matroid with rank function
\[
r_{\mathcal M^{\rho^*}}(S)=\min_{T\subseteq S}\{|S\setminus T|+\rho^*r_{\mathcal M}(T)\}.
\]
Since $\mathcal M$ has $\rho^*$ pairwise disjoint bases, $r_{\mathcal M^{\rho^*}}(N)=\rho^*r_{\mathcal M}(N)$, and the bases of $\mathcal M^{\rho^*}$ are precisely the members of $\mathcal B^{\rho^*}$.

We use the following observation. If $z\in[0,1]^N$ and $z/\rho^*$ is a convex combination of incidence vectors of bases of $\mathcal M$, then $z(N)=\rho^*r_{\mathcal M}(N)$ and, for $T\subseteq S\subseteq N$, $z(S)\leq z(T)+|S\setminus T|\leq\rho^*r_{\mathcal M}(T)+|S\setminus T|$. Taking the minimum over $T$ gives $z(S)\leq r_{\mathcal M^{\rho^*}}(S)$. The rank description of the base polytope~\cite[Corollary~40.2d]{Schr03} therefore shows that $z$ is a convex combination of the incidence vectors of sets in $\mathcal B^{\rho^*}$.

Apply this observation to $z=(\rho^*/\rho_f^*)c$. The capacity constraints give $z\in[0,1]^N$, and the first assertion gives the required convex representation of $z/\rho^*$. This proves the second assertion.
\end{proof}

The least core can be determined from the base constraints, with the excess parameter scaled as follows.
\begin{lemma}\label{fractional-base-reduction}
The linear program
\[
(\mathrm{LP})_B:\quad
\begin{array}{rl}
\mbox{max}&\varepsilon\\[1mm]
\mbox{s.t.}&\left\{
\begin{array}{ll}
x(B)\geq1+\varepsilon/\rho^*,&\quad B\in\mathcal B,\\[1mm]
x(N)=\rho^*,&\\[1mm]
x(v)\geq0,&\quad v\in N
\end{array}\right.
\end{array}
\]
has the same optimal allocations as $(\mathrm{LP})_1$, and their common optimal value is $\varepsilon_1^*=(\rho^*)^2/\rho_f^*-\rho^*$. Moreover, $x\in\mathcal{LC}(\Gamma_{\mathcal M})$ if and only if $(\rho_f^*/\rho^*)x$ is an optimal solution of $(\mathrm{FD})$.
\end{lemma}

\begin{proof}
Every feasible solution of $(\mathrm{LP})_B$ is feasible for $(\mathrm{LP})_1$. Indeed, summing its base constraints over a maximum base packing gives $\rho^*+\varepsilon\leq x(N)=\rho^*$, so $\varepsilon\leq0$. For $S\subseteq N$ with $\gamma(S)=k$, summing over $k$ pairwise disjoint bases in $S$ yields $x(S)-\gamma(S)\geq k\varepsilon/\rho^*\geq\varepsilon$.

Let $x\in X_1$. For every $K\in\mathcal B^{\rho^*}$, we have $x(K)\geq\rho^*+\varepsilon_1^*$. Taking the convex combination supplied by Lemma~\ref{fractional-load-decomposition}, and using $c(v)\leq1$, gives
\begin{equation}\label{eq:fractional-lc-bound}
\rho^*+\varepsilon_1^*
\leq\frac{\rho^*}{\rho_f^*}\sum_{v\in N}c(v)x(v)
\leq\frac{(\rho^*)^2}{\rho_f^*}.
\end{equation}
Conversely, if $y$ is an optimal solution of $(\mathrm{FD})$, then $x=(\rho^*/\rho_f^*)y$ is feasible for $(\mathrm{LP})_B$ at $\varepsilon=(\rho^*)^2/\rho_f^*-\rho^*$. Thus both programs attain this value, and every optimal allocation of $(\mathrm{LP})_B$ belongs to $X_1$.

To prove the reverse inclusion, let $x\in X_1$ and fix $B\in\mathcal B$. Both inequalities in~\eqref{eq:fractional-lc-bound} are equalities. Since $(\rho^*/\rho_f^*)c(v)<1$ for every $v$, for sufficiently small $\delta>0$ the vector $(1-\delta)(\rho^*/\rho_f^*)c+\delta\rho^*\chi^B$ belongs to $[0,1]^N$. Dividing it by $\rho^*$ gives a convex combination of incidence vectors of bases of $\mathcal M$. The rank argument in the proof of Lemma~\ref{fractional-load-decomposition} therefore gives a convex decomposition into incidence vectors of sets in $\mathcal B^{\rho^*}$. Taking their weights with respect to $x$ yields
\[
\rho^*+\varepsilon_1^*
\leq(1-\delta)(\rho^*+\varepsilon_1^*)+\delta\rho^*x(B).
\]
Hence $x(B)\geq1+\varepsilon_1^*/\rho^*$, proving equality of the optimal allocation sets.

At the common optimum, the base constraints are $x(B)\geq\rho^*/\rho_f^*$. After scaling by $\rho_f^*/\rho^*$, these are exactly the constraints of $(\mathrm{FD})$, and the total weight is $\rho_f^*$. This proves the final assertion.
\end{proof}

We now characterize the least core by the fixed exchange graph $G[\widetilde{\mathcal P}_B]$.

\begin{theorem}\label{LC-characterization}
Let $x\in X(\Gamma_{\mathcal M})$. Then $x\in\mathcal{LC}(\Gamma_{\mathcal M})$ if and only if
\begin{enumerate}
\item[(a)] $x(v)=0$ for $v\in\widetilde N_0$;
\item[(b)] $x(u)\leq x(v)$ for $(u,v)\in E[\widetilde{\mathcal P}_B]$.
\end{enumerate}
Furthermore, every $x\in\mathcal{LC}(\Gamma_{\mathcal M})$ satisfies
\begin{equation}\label{eq:fractional-minimum-bases}
x(\widetilde B_i)=\frac{\rho^*}{\rho_f^*}=1+\frac{\varepsilon_1^*}{\rho^*}
\quad(i=1,\ldots,s).
\end{equation}
\end{theorem}

\begin{proof}
Let $x\in\mathcal{LC}(\Gamma_{\mathcal M})$. By Lemma~\ref{fractional-base-reduction}, $y=(\rho_f^*/\rho^*)x$ is an optimal solution of $(\mathrm{FD})$. Complementary slackness with the fixed optimal fractional packing gives $y(v)(1-c(v))=0$ for $v\in N$ and $y(\widetilde B_i)=1$ for $i=1,\ldots,s$,
where the second equality uses $\lambda_i>0$. Thus $x(v)=0$ for $v\in\widetilde N_0$, proving (a). For $(u,v)\in E[\widetilde{\mathcal P}_B]$, choose $i$ such that $\widetilde B_i-u+v\in\mathcal B$. Dual feasibility gives $y(\widetilde B_i-u+v)\geq1=y(\widetilde B_i)$,
and hence $x(u)\leq x(v)$, proving (b). The equalities in~\eqref{eq:fractional-minimum-bases} follow from complementary slackness and Lemma~\ref{fractional-base-reduction}.

Conversely, let $x\in X(\Gamma_{\mathcal M})$ satisfy (a) and (b). For every $B\in\mathcal B$, Lemma~\ref{matroid exchange 1} gives a perfect matching from $\widetilde B_i\setminus B$ to $B\setminus\widetilde B_i$ in $E_{\widetilde B_i}$. Summing (b) over the matching shows that $x(B)\geq x(\widetilde B_i)$. Hence all $\widetilde B_i$ are minimum weight bases and have the same weight. By (a),
\[
\rho^*=x(N)=\sum_{v\in N}c(v)x(v)
=\sum_{i=1}^s\lambda_i x(\widetilde B_i)
=\rho_f^*x(\widetilde B_1).
\]
Therefore $x(\widetilde B_i)=\rho^*/\rho_f^*$ and $x(B)\geq\rho^*/\rho_f^*$ for every $B\in\mathcal B$. It follows that $(\rho_f^*/\rho^*)x$ is feasible for $(\mathrm{FD})$ with objective value $\rho_f^*$, and thus is optimal. Lemma~\ref{fractional-base-reduction} gives $x\in\mathcal{LC}(\Gamma_{\mathcal M})$.
\end{proof}

Recall from Section~\ref{sec:packing-exchange-graphs} that $\widetilde{\mathcal Q}$ is the collection of all SCCs of $G[\widetilde{\mathcal P}_B]$, and $\mathcal Q^c$ consists of the members from which there is no directed path to $\widetilde N_0$. For $Q,Q'\in\widetilde{\mathcal Q}$, write $Q\preceq Q'$ if and only if there is a directed path from $Q$ to $Q'$ in $G[\widetilde{\mathcal P}_B]$. Paths of length zero are allowed. Thus $\preceq$ is a partial order on $\widetilde{\mathcal Q}$, and we use the same notation for its restriction to $\mathcal Q^c$.

\begin{corollary}\label{fractional-canonical-partition}
Given the maximum fractional base packing $\widetilde{\mathcal P}_B$ and its exchange graph $G[\widetilde{\mathcal P}_B]$, if $x\in\mathcal{LC}(\Gamma_{\mathcal M})$, then
\begin{enumerate}
\item[(a)] If $u\in Q$, $v\in Q'$ and $Q\preceq Q'$, then $x(u)\leq x(v)$;
\item[(b)] If $u,v$ belong to the same SCC $Q\in\widetilde{\mathcal Q}$, then $x(u)=x(v)$. Henceforth we write $x_Q$ to denote the common value of the vertices in $Q$, i.e., $x_Q=\frac{x(Q)}{|Q|}$;
\item[(c)] If there is a directed path from $v$ to $\widetilde N_0$ in $G[\widetilde{\mathcal P}_B]$, then $x(v)=0$. Equivalently, $x(v)=0$ for every $v\in Q$ with $Q\in\widetilde{\mathcal Q}\setminus\mathcal Q^c$.
\end{enumerate}
Moreover, $\mathcal Q^c\ne\emptyset$.
\end{corollary}

\begin{proof}
Part (a) follows by applying Theorem~\ref{LC-characterization}(b) along a directed path, and (b) follows from mutual reachability. For (c), if $w\in\widetilde N_0$ is reachable from $v$, then Theorem~\ref{LC-characterization} gives $0\leq x(v)\leq x(w)=0$. Finally, the least core is nonempty by Lemma~\ref{fractional-base-reduction} and the existence of an optimal solution of $(\mathrm{FD})$. If $\mathcal Q^c=\emptyset$, part (c) would force every least core allocation to be zero on $N$, contradicting $x(N)=\rho^*>0$.
\end{proof}

\subsection{The nucleolus}

We next compute the nucleolus.

As in Section~4, define
$\widetilde N_*=\{v\in N:x(v)=0,\ \forall x\in X_1\}$ and
$\widetilde E_*=\{(u,v)\in E[\widetilde{\mathcal P}_B]:x(v)-x(u)=0,\ \forall x\in X_1\}$.
We next determine these two sets with respect to the least core. Consider the hierarchy partition $\mathcal Q_1^c\prec\cdots\prec\mathcal Q_t^c$ of $(\mathcal Q^c,\preceq)$ introduced in Section~2. For $i=1,\ldots,t$, set $n_i=\sum_{Q\in\mathcal Q_i^c}|Q|$ and $q=\sum_{i=1}^t n_i$.

As in~\eqref{core constru}, for each $Q\in\mathcal Q_j^c$, define $\widetilde x_Q=\alpha+(j-1)\beta$, where, if $t\geq2$, we choose $0<\alpha<\rho^*/q$ and $\beta=(\rho^*-q\alpha)/\sum_{j=1}^t(j-1)n_j$. If $t=1$, take $\alpha=\rho^*/q$ and $\beta=0$. Set $\widetilde x(v)=0$ for $v\in N\setminus\bigcup_{Q\in\mathcal Q^c}Q$. No edge goes from $\bigcup_{Q\in\mathcal Q^c}Q$ to its complement. Thus $\widetilde x(N)=\rho^*$, and $\widetilde x$ satisfies the zero coordinate and edge inequalities in Theorem~\ref{LC-characterization}. Hence $\widetilde x\in X_1$. Moreover, $\widetilde x_Q>0$ for every $Q\in\mathcal Q^c$, and $\widetilde x_Q<\widetilde x_{Q'}$ whenever $Q\prec Q'$ and $Q\ne Q'$. If $|\mathcal Q^c|\geq2$, then, for any $Q\in\mathcal Q^c$, decreasing $\widetilde x_Q$ by a sufficiently small $\delta>0$ and distributing the released total $|Q|\delta$ equally among the other SSC of $\mathcal Q^c$, uniformly within each SSC, yields another least core allocation.

\begin{lemma}\label{fractional-unfixed-local}
If $\mathcal{C}(\Gamma_{\mathcal M})=\emptyset$, then
\begin{enumerate}
\item[(a)] $v\in\widetilde N_*$ if and only if
$v\in N\setminus\bigcup_{Q\in\mathcal Q^c}Q$;
\item[(b)] For $(u,v)\in E[\widetilde{\mathcal P}_B]$, $(u,v)\in\widetilde E_*$ if and only if $u,v\in\widetilde N_*$ or $u,v$ belong to the same member of $\mathcal Q^c$.
\end{enumerate}
Moreover, if $|\mathcal Q^c|\geq2$, then $x(v)$ is not constant on $X_1$ for every $v\in N\setminus\widetilde N_*$, and $x(v)-x(u)$ is not constant on $X_1$ for every $(u,v)\in E[\widetilde{\mathcal P}_B]\setminus\widetilde E_*$.
\end{lemma}

If $|\mathcal Q^c|=1$, Theorem~\ref{LC-characterization} and Corollary~\ref{fractional-canonical-partition} imply that the least core consists of the single allocation
\[
\hat x(v)=
\begin{cases}
\dfrac{\rho^*}{|N\setminus\widetilde N_*|},&v\in N\setminus\widetilde N_*,\\[2mm]
0,&v\in\widetilde N_*.
\end{cases}
\]
This allocation is the nucleolus. In the remainder of this section, assume that $|\mathcal Q^c|\geq2$.

The second program in SLP$(\eta(\Gamma_{\mathcal M}))$ is
\[
(\textnormal{LP}_2):\quad
\begin{array}{rl}
\mbox{max}&\varepsilon\\[1mm]
\mbox{s.t.}&\left\{
\begin{array}{ll}
x(S)\geq\gamma(S)+\varepsilon,
&\forall S\in2^N\setminus\operatorname{fix}(X_1),\\[1mm]
x\in X_1.&
\end{array}\right.
\end{array}
\]
Consider the following linear program:
\[
(\textnormal{LP}''_2):\quad
\begin{array}{rl}
\mbox{max}&\varepsilon\\[1mm]
\mbox{s.t.}&\left\{
\begin{array}{ll}
x(v)-x(u)\geq\varepsilon-\varepsilon_1^*,
&(u,v)\in E[\widetilde{\mathcal P}_B]\setminus\widetilde E_*,\\[1mm]
x(v)\geq\varepsilon-\varepsilon_1^*,
&v\in N\setminus\widetilde N_*,\\[1mm]
x(v)-x(u)=0,
&(u,v)\in\widetilde E_*,\\[1mm]
x(v)=0,&v\in\widetilde N_*,\\[1mm]
x(N)=\rho^*.&
\end{array}\right.
\end{array}
\]
\begin{lemma}\label{ept nucleolus eq}
$(\textnormal{LP})''_2$ is equivalent to $(\textnormal{LP})_2$ in $SLP(\eta(\Gamma_{\mathcal M}))$, in the sense that they have the same optimal value and optimal allocations.
\end{lemma}

\begin{proof}
Let $(x,\varepsilon)$ be an optimal solution of $(\mathrm{LP}''_2)$. Since every allocation in $X_1$ is feasible for this program at $\varepsilon=\varepsilon_1^*$, optimality and the local constraints imply that the vertex values and edge differences are nonnegative. By Theorem~\ref{LC-characterization}, $x\in X_1$. We show that $(x,\varepsilon)$ is feasible for $(\mathrm{LP}_2)$. Consider $S\notin\operatorname{fix}(X_1)$, and choose $k=\gamma(S)$ pairwise disjoint bases contained in $S$.

For each of these bases $B$, choose a perfect matching $\pi$ from $\widetilde B_1\setminus B$ to $B\setminus\widetilde B_1$. By Lemma~\ref{matroid exchange 1} and~\eqref{eq:fractional-minimum-bases},
\[
x(B)=1+\frac{\varepsilon_1^*}{\rho^*}
+\sum_{(u,v)\in\pi}\bigl[x(v)-x(u)\bigr].
\]
If every arc in these matchings belonged to $\widetilde E_*$ and every element of $S$ outside the chosen bases belonged to $\widetilde N_*$, then $x(S)$ would be constant on $X_1$, a contradiction. Hence at least one matching arc is outside $\widetilde E_*$, or one of the remaining elements is outside $\widetilde N_*$.  Summing the weights of the disjoint bases and the remaining elements gives
\[
x(S)\geq k\left(1+\frac{\varepsilon_1^*}{\rho^*}\right)
+\varepsilon-\varepsilon_1^*\geq k+\varepsilon,
\]
where the last inequality follows from $k\leq\rho^*$ and $\varepsilon_1^*<0$. Thus $(x,\varepsilon)$ is feasible for $(\mathrm{LP}_2)$.

Conversely, let $(x,\varepsilon)$ be feasible for $(\mathrm{LP}_2)$. By Lemma~\ref{fractional-load-decomposition}, fix a convex decomposition of $(\rho^*/\rho_f^*)c$ into incidence vectors of sets in $\mathcal B^{\rho^*}$, all with positive coefficients. Equality in~\eqref{eq:fractional-lc-bound} and the least core constraints imply that every base $K$ in this decomposition satisfies $x(K)=\rho^*+\varepsilon_1^*$ throughout $X_1$.

For $v\in N\setminus\widetilde N_*$, since $(\rho^*/\rho_f^*)c(v)<1$, some base $K$ in the decomposition omits $v$. By Lemma~\ref{fractional-unfixed-local}, $x(v)$ is not constant on $X_1$, so $K+v\notin\operatorname{fix}(X_1)$. Its coalition constraint, together with $x(K)=\rho^*+\varepsilon_1^*$ and $\gamma(K+v)=\rho^*$, gives $x(v)\geq\varepsilon-\varepsilon_1^*$.

Now let $(u,v)\in E[\widetilde{\mathcal P}_B]\setminus\widetilde E_*$. If $u\in\widetilde N_*$, the required inequality follows from the vertex constraint proved above. Otherwise, let $Q_u,Q_v\in\mathcal Q^c$ be the components containing $u,v$, respectively. Then $Q_u\preceq Q_v$ and $Q_u\ne Q_v$. We may assume that no component lies strictly between $Q_u$ and $Q_v$ in the partial order; otherwise, the constraint follows along a chain of such pairs by monotonicity.

For every $R\in\mathcal Q^c$, distributing $\rho^*$ uniformly over all elements in the components $R'$ with $R\preceq R'$, and assigning zero elsewhere, gives an allocation in $X_1$ by Theorem~\ref{LC-characterization}. Let $y$ be a convex combination with positive coefficients of these allocations for all $R$ satisfying $Q_u\preceq R$ and $R\ne Q_v$. By construction, $y_{Q_u}=y_{Q_v}>0$, and every other component has value either zero or greater than $y_{Q_u}$.

For sufficiently small $\delta>0$, transfer weight $\delta/\rho^*$ from a positive weight base $\widetilde B_i$ inducing $(u,v)$ to $\widetilde B_i-u+v$ in the convex representation of $c/\rho_f^*$. Since $(\rho^*/\rho_f^*)c(v)<1$, the rank argument in Lemma~\ref{fractional-load-decomposition} shows that $(\rho^*/\rho_f^*)c+\delta(\chi^{\{v\}}-\chi^{\{u\}})$ is a convex combination of incidence vectors of sets in $\mathcal B^{\rho^*}$. Since $y(u)=y(v)$, its weight with respect to $y$ remains $\rho^*+\varepsilon_1^*$, so every set $L$ with positive coefficient in this new decomposition satisfies $y(L)=\rho^*+\varepsilon_1^*$. The perturbation decreases the average intersection with $Q_u$ by $\delta$. Hence some $K$ in the fixed decomposition and some $L$ in the new decomposition satisfy $|K\cap Q_u|>|L\cap Q_u|$.

Apply Lemma~\ref{matroid exchange 1} in $\mathcal M^{\rho^*}$ to $K$ and $L$. Each matched pair $(a,b)$ satisfies $K-a+b\in\mathcal B^{\rho^*}$ and $y(b)-y(a)\geq0$. These differences sum to $y(L)-y(K)=0$, so they are all zero. Since $|K\cap Q_u|>|L\cap Q_u|$, at least one matched pair has $a\in Q_u$ and $b\notin Q_u$. Our choice of $y$ forces $b\in Q_v$. The set $K$ is tight throughout $X_1$, whereas $x(b)-x(a)=x(v)-x(u)$ is not constant there by Lemma~\ref{fractional-unfixed-local}. Thus $K-a+b\notin\operatorname{fix}(X_1)$, and its coalition constraint gives $x(v)-x(u)\geq\varepsilon-\varepsilon_1^*$.

The remaining constraints of $(\mathrm{LP}''_2)$ follow from $x\in X_1$. Thus every feasible solution of $(\mathrm{LP}_2)$ is feasible for $(\mathrm{LP}''_2)$, and every optimal solution of $(\mathrm{LP}''_2)$ is feasible for $(\mathrm{LP}_2)$. Therefore the two programs have the same optimal value and the same optimal allocations.
\end{proof}

Using the hierarchy partition and the numbers $n_i$ defined above, we obtain the following result.

\begin{lemma}\label{nucleolus21}
The unique optimal allocation of $(\mathrm{LP}''_2)$ is given by
\begin{equation}\label{eq:fractional-nucleolus}
x_Q=i(\varepsilon_2^*-\varepsilon_1^*)
\quad(Q\in\mathcal Q_i^c),\qquad
\varepsilon_2^*=\varepsilon_1^*
+\frac{\rho^*}{n_1+2n_2+\cdots+tn_t},
\end{equation}
and $x(v)=0$ for $v\in\widetilde N_*$.
\end{lemma}

\begin{proof}
The proof follows the same hierarchy argument as that of Lemma~\ref{nucleolus2}, with $\bar{\mathcal Q}_i^+$ and $\varepsilon$ replaced by $\mathcal Q_i^c$ and $\varepsilon-\varepsilon_1^*$, respectively.
\end{proof}

By Lemmas~\ref{ept nucleolus eq} and~\ref{nucleolus21}, the second program in SLP$(\eta(\Gamma_{\mathcal M}))$ has a unique optimal allocation. This allocation is the nucleolus.

\begin{theorem}\label{fractional-nucleolus-algorithm}
Let $\Gamma_{\mathcal M}=(N,\gamma)$ be the matroid packing game defined on $\mathcal M=(N,\mathcal I)$. When $\mathcal C(\Gamma_{\mathcal M})=\emptyset$, the least core value, the partition $\mathcal Q^c$, its partial order and the nucleolus $\eta(\Gamma_{\mathcal M})$ can be computed in polynomial time.
\end{theorem}

\begin{proof}
Compute a maximum fractional base packing and its positive weight bases using the algorithm recalled in Section~\ref{sec:packing-preliminaries}. For each $\widetilde B_i$, test $\widetilde B_i-u+v$ for all $u\in\widetilde B_i$ and $v\notin\widetilde B_i$. This constructs $G[\widetilde{\mathcal P}_B]$ using at most $s\,r_{\mathcal M}(N)\bigl(|N|-r_{\mathcal M}(N)\bigr)$
independence oracle calls and $O(s|N|^2)$ additional operations. Compute $\widetilde N_0$ from the loads, and obtain $\widetilde N_*$ by a reverse search from $\widetilde N_0$. Then compute the SCCs on $N\setminus\widetilde N_*$ and their condensation graph. These operations take $O(|N|+|E[\widetilde{\mathcal P}_B]|)$ time. The condensation graph represents $\preceq$; its transitive closure, if required, can be computed in $O(|N|(|N|+|E[\widetilde{\mathcal P}_B]|))$ time. The hierarchy partition and the nucleolus require a topological traversal and summation, taking $O(|N|+|E[\widetilde{\mathcal P}_B]|)$ additional time.

The support size $s$ is polynomial in $|N|$. The packing number $\rho^*$ can also be computed in polynomial time, as recalled in Section~\ref{sec:packing-preliminaries}. Therefore all the stated computations are polynomial in the independence oracle model.
\end{proof}

\section{Conclusion}
For future work, one direction is to study other solution concepts for matroid
packing games. We have obtained results for the nucleon, a multiplicative variant
of the nucleolus, and shown that it coincides either with a core allocation or
with the nucleolus.

Another direction is matroid covering games, which generalize arboricity games
from graphic matroids to general matroids. Using the framework developed in
this paper, we have obtained characterizations of the core and polynomial time
algorithms for computing the nucleolus. We are currently investigating whether matroid packing and covering games can be
studied under a unified theoretical framework.
A recent work~\cite{EE26} independently resolves the open problem left by the
2020 study of network strength games~\cite{BB20}. Our work <reaches this conclusion from a
different perspective: network strength games are the graphic matroid special
case of matroid packing games. Complementary to their approach, our method avoids a general ellipsoid algorithm
framework and yields both a polynomial time algorithm and a structural
characterization of the nucleolus.
A recent work~\cite{EE26} proposes a new framework for computing the nucleolus and proves that the nucleolus of network strength games can be computed in polynomial time. The approach operates in a reduced optimization space induced by a set of unfixed coalitions. The method  relies on the ellipsoid algorithm and assumes that the family of fixed coalitions is given in advance.
 Our work reaches this conclusion from a
different perspective: network strength games are the graphic matroid special
case of matroid packing games. Complementary to their approach, our method avoids a general ellipsoid algorithm
framework and yields both a polynomial time algorithm and a structural
characterization of the nucleolus.

\bibliographystyle{splncs04}
\bibliography{wine627}

%
%
%
%
\end{document}